\documentclass[journal, letterpaper]{IEEEtran}
\IEEEoverridecommandlockouts

\PassOptionsToPackage{dvipsnames}{xcolor}

\usepackage{pgfplots}
\pdfoutput=1

\usepackage[utf8]{inputenc}
\usepackage{amsmath}
\usepackage{amsfonts}
\usepackage{amssymb}
\usepackage{amsthm}

\usepackage{algorithm2e}
\usepackage{capt-of}

\usepackage{pgfplotstable}
\usepackage{comment}
\usepackage{graphicx}
\usepackage{sidecap} 
\usepackage{physics}
\usepackage{mathtools}
\usepackage{dsfont}
\pgfplotsset{compat=1.18}
\usepackage{tikz}

\usepackage[colorlinks]{hyperref}
\hypersetup{
    colorlinks  = true,
    citecolor   = PineGreen,
    linkcolor   = MidnightBlue,
    urlcolor    = PineGreen
}

\usepackage[capitalise]{cleveref}

\graphicspath{{figures/}{./}}

\usepackage{thmtools}
\declaretheorem[parent=section]{theorem}

\declaretheorem[numberlike=theorem, name=Lemma]{lemma}

\def\tr{\operatorname{tr}}

\def\ket#1{|#1\rangle}

\def\c2{\ensuremath{C^{(2)}}}

\allowdisplaybreaks

\pgfplotstableread{
lam                        eps                         r
239.999999997135          0.00175494271616117         312
183.15789473691           0.00168098203758404         273
148.085106382963          0.00160696511264291         245
124.285714285717          0.00153294796284342         225
107.076923076919          0.00145892695059047         209
94.0540540540426          0.00138491369123983         196
83.8554216867482          0.0013108965505495          185
75.6521739130454          0.0012368794102855          176
68.9108910891046          0.00116286226896989         168
63.272727272725           0.00108884284691158         161
58.4873949579805          0.00101482798835661         155
54.3750000000008          0.000940807259219412        149
50.8029197080297          0.000866793706585156        144
47.6712328767124          0.000792775288033454        140
44.9032258064504          0.000718758260525476        136
42.439024390244           0.000644742281934896        132
40.2312138728327          0.000570725121825433        128
38.2417582417587          0.000496708003765889        125
36.4397905759146          0.000422690857007524        122
34.8000000000002          0.000348673721056514        120
33.3014354066992          0.000274656581048305        117
31.9266055045872          0.000200639429898786        115
30.6607929515419          0.000126622301984014        112
29.4915254237285          5.26051593610077e-05        110
28.4081632652908          2.09496957239708e-05        108
27.4015748031434          9.33688534452415e-05        106
26.4638783269942          0.000165788005993761        104
25.5882352941165          0.000238207151028291        103
24.7686832740202          0.000310626294325544        101
24                         0.000383045443996366        99
}\datatable

\pgfplotstableread{
x                       y
0.01                    2.755737054959
0.0131034482758621      2.774436927262
0.0162068965517241      2.793993551707
0.0193103448275862      2.814472587372
0.0224137931034483      2.835966452966
0.0255172413793103      2.858577291283
0.0286206896551724      2.882431579323
0.0317241379310345      2.907672639902
0.0348275862068966      2.934471720670
0.0379310344827586      2.963034797603
0.0410344827586207      2.993607563700
0.0441379310344828      3.026499340262
0.0472413793103448      3.062084250565
0.0503448275862069      3.100849895907
0.053448275862069       3.143417151026
0.056551724137931       3.190613848001
0.0596551724137931      3.243573010524
0.0627586206896552      3.303898841891
0.0658620689655172      3.373977146053
0.0689655172413793      3.457580783425
0.0720689655172414      3.561209990385
0.07517241379310351     3.697583714861
0.0782758620689655      3.897489795301
0.08137931034482759     4.278971659418
0.0844827586206897      4.678822280400
0.08758620689655169     4.029797974228
0.0906896551724138      3.780446891875
0.09379310344827591     3.623045205047
0.0968965517241379      3.507761784290
0.1                     3.416749698861
}\datatableB

\pgfplotstableset{
  create on use/x/.style={
    create col/expr={sqrt(\thisrow{lam}/\thisrow{eps})}
  }
}

\begin{document}

\title{Semidefinite optimization of the\\quantum relative entropy of channels}
\date{\today}

\author{Gereon Ko{\ss}mann and Mark M. Wilde \thanks{Gereon Ko{\ss}mann is with the Institute for Quantum Information, RWTH Aachen University, Aachen, Germany (email: kossmann@physik.rwth-aachen.de), and Mark M. Wilde is with the School of Electrical and Computer Engineering, Cornell University, Ithaca, New York 14850, USA (email: wilde@cornell.edu).}}

\maketitle

\begin{abstract}
    This paper introduces a method for calculating the quantum relative entropy of channels, an essential quantity in quantum channel discrimination and resource theories of quantum channels. By building on recent developments in the optimization of relative entropy for quantum states [Ko{\ss}mann and Schwonnek, arXiv:2404.17016], we introduce a discretized linearization of the integral representation for the relative entropy of states, enabling us to handle maximization tasks for the relative entropy of channels. Our approach here extends previous work on minimizing relative entropy to the more complicated domain of maximization. It also provides efficiently computable upper and lower bounds that sandwich the true value with any desired precision, leading to a practical method for computing the relative entropy of channels.   
\end{abstract}
 
\tableofcontents

\section{Introduction}

Relative entropy is a well-established concept in both  classical and quantum information theory~\cite{Tomamichel_2016,hayashi2017QuantumInformationTheory,wilde2017QuantumInformationTheory,watrous2018TheoryQuantumInformation,HOLEVO-CHANNELS-2,khatriwilde2024}, where it  measures the deviation or divergence between two probability distributions in the classical case or between  two quantum states in the quantum case; it has a precise operational meaning in quantum hypothesis testing as the optimal exponential rate at which the type~II error probability decays to zero if there is a constant constraint on the type~I error probability~\cite{hiai1991ProperFormulaRelative,nagaoka2000StrongConverseSteins}.  The relative entropy of two quantum states $\rho,\sigma \in \mathcal{S}(\mathcal{H})$, i.e., positive semidefinite operators with trace equal to one, is defined as~\cite{umegaki1962ConditionalExpectationOperator}
\begin{equation}
\label{eq:def_relative_entropy}
    D(\rho \Vert \sigma) \coloneqq  \begin{cases}
        \tr[\rho (\ln\rho -  \ln\sigma)] & \ker[\sigma] \subseteq \ker[\rho] \\
        + \infty & \text{else}.
    \end{cases}
\end{equation}
More generally, we use this same formula to define quantum relative entropy when $\rho $ and $\sigma$ are positive semidefinite (i.e., not necessarily having a trace constraint).

Going beyond the setting of quantum hypothesis testing, quantum relative entropy has found applications in a variety of scenarios, most prominently in quantifying entanglement in a bipartite state~\cite{VPRK97,VP98,Rai99,Rai01} and more generally in quantifying the resource-theoretic value of quantum states~\cite{Chitambar_2019}. For these problems, it is necessary to optimize the relative entropy over a class of states, and consequently this has generated broad interest in the field of relative entropy optimization. After a number of important  contributions in this domain~\cite{Fawzi2017,Fawzi2018,Fawzi2019,fawzi2023optimalselfconcordantbarriersquantum,faust2023rationalapproximationsoperatormonotone}, the most recent work of~\cite{koßmann2024optimisingrelativeentropysemi} provided a compelling method for optimizing the relative entropy by making use of a recent integral representation for it~\cite{Frenkel_2023,jenčová2023recoverabilityquantumchannelshypothesis}. Before proceeding, let us indicate  that~\cite{koßmann2024optimisingrelativeentropysemi} is the paper most closely related to our results presented here and upon which we build to develop them.

Our work here goes beyond optimizing the relative entropy of states, with our main contribution being to provide methods for optimizing the relative entropy of channels. Before stating our findings more precisely, let us first motivate why this is important in the context of quantum information theory.
In this realm, quantifying  resources is a central theme. When discussing resources, we refer to entities or operations that hold some kind of value necessary for carrying out specific tasks. The field of resource theories in quantum information has developed to address this concern by providing structured ways to categorize and measure resources, especially when dealing with quantum systems. Quantum resource theories serve to analyze which operations or states are valuable, how they can be manipulated, and how they are constrained under different physical principles~\cite{Chitambar_2019}.
Our results here find applications in  resource theories of quantum channels~\cite{CFS16,TEZP19,liu2019resourcetheoriesquantumchannels,LY20}, due to the prominent role that the relative entropy of channels plays in such resource theories. Channels, in this context, are operations that manipulate quantum states in the most general physically allowed manner and thus can transform the resource-theoretic value of a state. One of the key tasks in resource theories of quantum channels is to quantify the value of these resources~\cite{Cooney_2016,CFS16,TEZP19,liu2019resourcetheoriesquantumchannels,LY20,Gour_2021}.

A powerful and central tool for quantifying the resource-theoretic value of channels is the relative entropy of channels~\cite{Cooney_2016,Leditzky_2018}, which serves as a measure of distinguishability of two quantum channels. 
For  quantum channels $\mathcal{N}_{A\to B}$ and $\mathcal{M}_{A\to B}$, it is defined as~\cite[Eq.~(1.14)]{Cooney_2016}
\begin{align}\label{eq:problem_statement_channel_relative}
    D(\mathcal{N} \Vert \mathcal{M}) \coloneqq  \sup_{\rho_{AR}} \ D(\mathcal{N}_{A\to B}(\rho_{AR})\Vert \mathcal{M}_{A\to B}(\rho_{AR})),
\end{align}
where the supremum is over every bipartite state $\rho_{AR}$ and also over the dimension of the reference system $R$. More generally, we use this same formula if $\mathcal{N}_{A\to B}$ and $\mathcal{M}_{A\to B}$ are completely positive maps (i.e., not necessarily having a trace-preservation constraint).
In the context of quantum channels, relative entropy quantifies how much one channel deviates from another, particularly when comparing a resourceful channel to a free channel; additionally, it finds a precise operational meaning in the context of quantum channel discrimination~\cite{Cooney_2016,Wilde_2020,WW19}. 
This makes it an essential quantity in tasks in many quantum information processes beyond channel discrimination, and in other tasks related to resource theories of channels, like error correction, quantum communication, and quantum cryptography~\cite{Metger_2022}, where understanding how much noise or disturbance a channel introduces is crucial. 

The relative entropy of channels, however, is difficult to compute directly due to the inherent complexity of quantum channels and the relative entropy of states itself. This paper addresses these challenges by developing methods to optimize the quantum relative entropy of channels, by using techniques from semidefinite programming~\cite{BV04,Skrzypczyk_2023} (see~\cite{Fang_2021,fawzi2022semidefiniteprogramminglowerbounds,fawzi2023optimalselfconcordantbarriersquantum,brown2023deviceindependentlowerboundsconditional, koßmann2024optimisingrelativeentropysemi,huang2024semidefiniteoptimizationmeasuredrelative} for various examples in which it has played a role in optimizing entropic functions of quantum states and channels).
A semidefinite program (SDP) is a form of optimization that is particularly suited to problems involving quantum systems, as it deals with constraints on matrices that are positive semidefinite — a common condition in quantum information theory due to the basic postulates of quantum mechanics. More concretely, 
in this work we address the optimization in~\eqref{eq:problem_statement_channel_relative} with the help of a recently developed integral representation from~\cite{Frenkel_2023} (see also~\cite{jenčová2023recoverabilityquantumchannelshypothesis,Hirche_2024}).  Then we develop, with the help of~\cite{koßmann2024optimisingrelativeentropysemi}, a method for maximizing a discretized version of the integral representation for the particular case of~\eqref{eq:problem_statement_channel_relative}. 

Our paper overcomes key technical obstacles to optimizing the relative entropy of channels. First, prior to the appearance of~\cite{koßmann2024optimisingrelativeentropysemi}, one of the main methods employed for optimizing the relative entropy of states was given by~\cite{Fawzi2017,Fawzi2018,Fawzi2019}. However, this approach employs a ``lifting'' technique, which doubles the Hilbert space of the states involved, whereby the relative entropy can be represented as
\begin{align}
D(\rho \Vert \sigma) & = \langle \Gamma | D_{\operatorname{op}}(\rho \otimes I \Vert I \otimes \sigma^T) |\Gamma \rangle, \\
D_{\operatorname{op}}(A\Vert B) & \coloneqq - A^{1/2} \ln (A^{-1/2} B A^{-1/2}) A^{1/2},
\end{align}
and $|\Gamma \rangle$ is defined later on in~\eqref{eq:Gamma-vec}.  While this approach is helpful for optimizing the relative entropy of states, it has been unclear for some time now how to employ it for optimizing the relative entropy of channels. As a second obstacle, observe that here we are considering a maximization task involving relative entropy, rather than a minimization task as considered in prior work, and the relative entropy of states is furthermore a nonlinear and convex function of the input states. While our results are presented in much greater detail later on, let us indicate here that we overcome the first obstacle by employing the integral representation of the relative entropy from~\cite{Frenkel_2023,Hirche_2024}, as done in~\cite{koßmann2024optimisingrelativeentropysemi}, which avoids the use of the lifting technique that doubles the Hilbert space. We overcome the second obstacle by appealing to the facts that 1) the relative entropy of channels is actually a concave function of the reduced density operator that represents the full channel input state (\cite[Prop.~13]{Wang2019magic} and \cite[Prop.~7.83]{khatriwilde2024}) and 2) the discretized linearization of the integral representation from~\cite{koßmann2024optimisingrelativeentropysemi} is simpler to handle analytically.

Next to the technical contributions in terms of executable programs, we apply our tools to general problems having the following structure:
\begin{align}\label{eq:optimization_free_states}
    \inf_{\mathcal{M} \in \mathcal{F}} D(\mathcal{N}\Vert \mathcal{M}),     
\end{align}
where $\mathcal{F}$ is a set of free channels (or free completely positive maps more generally). As already mentioned, given that the relative entropy of channels is a natural measure for quantifying the distinguishability of quantum channels, the quantity~\eqref{eq:optimization_free_states} can be seen as the deviation of a generally non-free channel~$\mathcal{N}$ from the set of free channels in a given resource theory. Related to the fact that the relative entropy of channels is a resource monotone, this optimization problem is related to the randomness required to erase the resources contained in a channel $\mathcal{N}$~\cite[Sections~VI and~VII]{liu2019resourcetheoriesquantumchannels}.

Let us briefly recall here some prior work on algorithms for computing other kinds of channel relative entropies besides the one in~\eqref{eq:problem_statement_channel_relative}. In~\cite{Fang_2021}, an algorithm was given for computing the geometric R\'enyi relative entropies of channels, which approximate the Belavkin--Staszewski relative entropy of channels, and in~\cite{huang2024semidefiniteoptimizationmeasuredrelative}, an algorithm was given for computing the measured relative entropy and measured R\'enyi relative entropies of channels. The relative entropy of channels considered here is arguably the most prominent channel relative entropy to consider, due to its various applications mentioned in Section~\ref{sec:application_resource_theory}. However, we should note that one can also argue for the significance of the measured relative entropies due to the limited capabilities of present-day quantum devices, and additionally for the Belavkin--Staszewski relative entropy of channels, due to its role in providing upper bounds on optimal rates for various tasks.

The rest of our paper is organized as follows. In \cref{sec:notation}, we establish preliminary concepts and notation that are used throughout. In \cref{sec:main_results}, we present our main result and derivations for optimizing the relative entropy of channels; see \cref{sec:lower-bound} for the lower bound and \cref{sec:upper-bound} for the upper bound. In \cref{subsec:recap_resource_theory_channels}, we provide a brief introduction to resource theories of quantum channels and indicate how our findings are relevant in these contexts. In \cref{subsec:optimization_task}, we integrate techniques for optimizing the relative entropy of channels with a concrete quantum resource theory to derive bounds on the asymptotic rates of a channel-simulation task under that resource theory. In \cref{subsec:numerical_evidence}, we present an explicit numerical example and elaborate upon the numerical accuracy claimed in \cref{sec:main_results}. In \cref{subsec:k-extendible}, we outline how to apply our method to the special class of $k$-extendible channels and, together with a de Finetti argument, obtain bounds on the relative entropy of channels when the free channels in the resource theory are entanglement-breaking channels. 
%\mmw{include summary of Sections IV-B and V}
%In \cref{sec:examples}, we showcase our optimization findings on several key examples of interest.
Finally, in \cref{sec:conclusion}, we conclude with a summary and indicate some directions for future research.

\section{Preliminaries and notation}

\label{sec:notation}

In this paper we consider finite-dimensional Hilbert spaces denoted by $\mathcal{H}$. Systems are usually described with subscripts, as done with the notation $\mathcal{H}_A$. The set of states acting on a Hilbert space $\mathcal{H}$ is denoted by $\mathcal{S}(\mathcal{H})$ and abbreviates the set of positive semidefinite operators satisfying $\tr[\cdot] = 1$ and defined with respect to the usual cone $\mathcal{P}(\mathcal{H})$ of positive semidefinite operators acting on $\mathcal{H}$. %The positive semidefinite cone is denoted by .

Quantum channels are completely positive and trace preserving maps, denoted by 
\begin{align}
    \mathcal{N}_{A\to B}\colon \mathcal{S}(\mathcal{H}_A) \to \mathcal{S}(\mathcal{H}_B).
\end{align}
To each completely positive map $\mathcal{M}_{A\to B}$, we assign the essentially unique Choi matrix: 
\begin{align}\label{eq:choi_isomorphism}
    \Gamma_{AB}^\mathcal{M} \coloneqq (\operatorname{id}_A \otimes \mathcal{M}_{A\to B}) (\Gamma_{AA}) \in \mathcal{P}(\mathcal{H}_A \otimes \mathcal{H}_B).
\end{align}
The uniqueness is due to the  Choi--Jamiołkowski isomorphism.
In~\eqref{eq:choi_isomorphism}, the operator $\Gamma_{AA} \in \mathcal{S}(\mathcal{H}_A \otimes \mathcal{H}_A)$ is the maximally entangled operator, defined as
\begin{align}
    \Gamma_{AA} & \coloneqq | \Gamma \rangle \!\langle \Gamma |_{AA},\\
    | \Gamma \rangle & \coloneqq \sum_{i} \ket{ii}_{AA}.
    \label{eq:Gamma-vec}
\end{align}

For two states $\rho,\sigma \in \mathcal{S}(\mathcal{H})$, the quantum relative entropy is defined as in~\eqref{eq:def_relative_entropy}. Importantly, in all our calculations we use the natural logarithm to be consistent with the integral representations presented later on.
For a pair of positive semidefinite operators $\rho,\sigma \in \mathcal{P}(\mathcal{H})$, we define the variable
\begin{equation}
\ln\lambda \coloneqq  D_{\max}(\rho \Vert \sigma)    
\end{equation}
to be the max-relative entropy~\cite{datta2009MinMaxrelativeEntropies}, where
\begin{equation}
    D_{\max}(\rho \Vert \sigma) \coloneqq \ln \inf_{\gamma \geq 0} \{ \gamma : \rho \leq \gamma \sigma\}.\label{eq:dmax-def}
\end{equation}
The inequality in~\eqref{eq:dmax-def} should be understood in terms of the usual operator order; i.e., $A \geq B$ for Hermitian $A$ and $B$ if and only if $A - B \in \mathcal{P}(\mathcal{H})$. 
Furthermore, for a Hermitian operator $H \in \mathcal{B}(\mathcal{H})$ with unique Jordan--Hahn decomposition $H = H_+ - H_-$ (into orthogonal positive and negative parts), we define
\begin{equation}
\tr^+[H]\coloneqq  \tr[H_+].
\end{equation}
In our paper, we make extensive use of the following integral representation for the relative entropy~\cite{jenčová2023recoverabilityquantumchannelshypothesis}:
\begin{align}\label{eq:jencova_representation}
    D(\rho \Vert \sigma) = \tr[\rho - \sigma] + \int_\mu^\lambda \frac{ds}{s} \tr^+ [\sigma s - \rho] + \ln\lambda +1 - \lambda,
\end{align}
which is based on an integral representation established in~\cite{Frenkel_2023} and for which an alternate proof has recently been given in~\cite{liu2025layercakerepresentationsquantum}.
Observe that the first term $\tr[\rho - \sigma] = 0 $ if $\rho$ and $\sigma$ are both states.
The constants $\lambda,\mu \in \mathbb{R}_{\geq 0}$ are defined in terms of the following operator inequalities:
\begin{align}
    \mu \sigma \leq \rho \leq \lambda \sigma.
    \label{eq:lam-mu-rho-sig}
\end{align}
% From this we can immediately deduce that $D(\rho \Vert \sigma) = +\infty$  if and only if $\lambda = +\infty$.
Note that the smallest possible value of $\lambda $ satisfies $\ln \lambda = D_{\max}(\rho \Vert \sigma)$ and the largest possible value of $\mu$ satisfies $-\ln \mu = D_{\max}(\sigma \Vert \rho)$. From this we can immediately deduce that $D(\rho \Vert \sigma) = +\infty$  if and only if $D_{\max}(\rho \Vert \sigma) = +\infty$.

\section{Main result: Optimizing the quantum relative entropy of channels}

\label{sec:main_results}

The motivation for our paper stems from the study of the following optimization task involving quantum channels.  Let $\mathcal{N}_{A \to B}$ and $\mathcal{M}_{A\to B}$ be two quantum channels, where each channel maps states from a finite-dimensional Hilbert space $\mathcal{H}_A$ (the input system) to those of another finite-dimensional Hilbert space $\mathcal{H}_B$ (the output system). The quantum relative entropy of channels is a measure of distinguishability between these two channels and is formally defined as~\cite{Cooney_2016,Leditzky_2018,Gour_2021}: 
\begin{equation}
    D(\mathcal{N} \Vert \mathcal{M}) \coloneqq  \sup_{\rho_{AR}}  D(\mathcal{N}_{A\to B}(\rho_{AR})\Vert \mathcal{M}_{A\to B}(\rho_{AR})),
\end{equation}
where the optimization is performed over every quantum state $\rho_{AR}$, which is a density matrix on the joint system composed of $\mathcal{H}_A$ and an arbitrary reference system $\mathcal{H}_R$. Here, $R$ represents an ancillary reference system that can have an arbitrarily large dimension, and this dimension is thus also involved in the optimization. The quantity $D(\mathcal{N}_{A\to B}(\rho_{AR})\Vert \mathcal{M}_{A\to B}(\rho_{AR}))$ is the quantum relative entropy between the output states of the channels $\mathcal{N}_{A\to B}$ and $\mathcal{M}_{A\to B}$ after they act on the input state $\rho_{AR}$. As mentioned previously, more generally we use the same formula above when $\mathcal{N}_{A\to B}$ and $\mathcal{M}_{A\to B}$ are completely positive maps.

The operational significance of the relative entropy of channels, particularly in the context of quantum resource theories, is reviewed in depth in \cref{sec:application_resource_theory}. From a mathematical perspective, one can simplify the optimization over arbitrarily large reference systems. By considering a purification of the marginal state $\rho_A$ (the reduced state on system $A$), it can be shown that the optimal value when  optimizing over every possible state $\rho_{AR}$ is equal to the optimal value when optimizing over pure states on the enlarged Hilbert space $\mathcal{H}_A \otimes \mathcal{H}_{R}$, where $\mathcal{H}_{R}$ is isomorphic to $\mathcal{H}_A$. This reduction in the optimization problem was proven in~\cite[Lemma~7]{Cooney_2016} (see also~\cite[Definition~II.2]{Leditzky_2018} and the text thereafter) and is recalled  in \cref{lem:optimization_restriction}.

Additionally, applying~\cite[Lemma~7]{Cooney_2016} or~\cite[Prop.~7.82]{khatriwilde2024} (recalled here as \cref{lem:conic_formulation_relative_channels}), we can further reformulate the problem using the Choi representations of the channels $\mathcal{N}_{A\to B}$ and $\mathcal{M}_{A \to B}$ from~\eqref{eq:choi_isomorphism}. The Choi matrix of a quantum channel captures all the information about the channel and provides a convenient representation for optimization. In particular, the quantum relative entropy of channels can be recast as an optimization problem over the complete cone of quantum states on the isomorphic system~$\mathcal{H}_{A}$:
\begin{equation}\label{eq:final_optimization_problem}
\begin{aligned}
         &D(\mathcal{N}\Vert \mathcal{M}) = &\sup \ &D\!\left(\rho_{A}^{1/2} \Gamma^\mathcal{N}_{A B}\rho_{A}^{1/2} \middle \Vert \rho_{A}^{1/2}\Gamma^\mathcal{M}_{A B}\rho_{A}^{1/2}\right)&  \\
         & &\operatorname{s.t.} \ &\rho_{A} \in \mathcal{S}(\mathcal{H}_{A}),&
    \end{aligned}
\end{equation}
where $\Gamma^{\mathcal{N}}_{A B}$ and $\Gamma^{\mathcal{M}}_{A B}$ are the Choi matrices associated with the channels $\mathcal{N}_{A \to B}$ and $\mathcal{M}_{A \to B}$, respectively. The optimization is carried out over the set of valid quantum states  on $\mathcal{H}_{A}$, each of which is denoted by $\rho_{A}$. It is known that this optimization problem is concave in $\rho_{A}$ (see, e.g., \cite[Prop.~13]{Wang2019magic} and \cite[Prop.~7.83]{khatriwilde2024}).

In the subsequent sections, we employ techniques from the fields of quantum information theory and mathematical optimization to derive upper and lower bounds for this optimization problem. In particular, we draw on methods from the recent works~\cite{jenčová2023recoverabilityquantumchannelshypothesis,Frenkel_2023,koßmann2024optimisingrelativeentropysemi} to construct these bounds. Moreover, for added generality, we impose an energy constraint on the reduced state of the system $\mathcal{H}_{A}$, which is a common consideration in physically relevant quantum systems. This constraint further refines the optimization problem and ensures that the resulting bounds are physically meaningful in realistic scenarios where energy constraints are imposed.

\subsection{Preliminary comments on optimizing channel relative entropy}

 The optimization problem~\eqref{eq:final_optimization_problem} can now be converted with the help of the integral representation~\eqref{eq:jencova_representation} into a form to which the  methods of~\cite{koßmann2024optimisingrelativeentropysemi} are applicable. However, here we are considering a maximization problem and not a minimization problem, the latter being done previously in~\cite{koßmann2024optimisingrelativeentropysemi}. As it turns out, in the particular case of the relative entropy of channels, maximizing does not make a large difference in the structure of the resulting optimization problem, because the relative entropy of channels is concave (\cite[Prop.~13]{Wang2019magic} and \cite[Prop.~7.83]{khatriwilde2024}). 
 %even though the relative entropy itself is of course non-linear and convex.
 
 The main idea of the approach is to discretize the integral from~\eqref{eq:jencova_representation} and to linearize the domain piecewise.  
 For the domain of the integral, given as the (compact) interval $[\mu,\lambda]$ in~\eqref{eq:jencova_representation}, we need a bound on $\lambda$, which is related to the max-relative entropy in~\eqref{eq:final_optimization_problem}, as previously mentioned after~\eqref{eq:lam-mu-rho-sig}. Importantly, the relative entropy is finite if and only if $\lambda$ is finite and thus if and only if $[\mu,\lambda]$ is compact (in the finite-dimensional case). To conclude, we have to calculate the max-relative entropy of the underlying problem class beforehand.
 
 As we will see now (and as observed previously in~\cite[Remark~19]{Wilde_2020}), the problem in~\eqref{eq:problem_statement_channel_relative} has a finite value if and only if the following optimization  
\begin{equation}
    \label{eq:max_relative_entropy}
\begin{aligned}
   &\ln\lambda \coloneqq  &\sup \ &D_{\max}\!\left(\rho_{A}^{1/2}\Gamma^\mathcal{N}_{A B}\rho_{A}^{1/2} \middle \Vert \rho_{A}^{1/2}\Gamma^\mathcal{M}_{A B}\rho_{A}^{1/2}\right) & \\ 
        & &\operatorname{s.t.} \ &\rho_{A} \in \mathcal{S}(\mathcal{H}_{A})&
\end{aligned}
\end{equation}
leads to a finite value. However, as observed in~\cite[Lemma~12]{Wilde_2020}, the optimization problem in~\eqref{eq:max_relative_entropy} can be reduced to $\ln\lambda = D_{\max}(\Gamma^\mathcal{N}_{A B} \Vert \Gamma^\mathcal{M}_{A B})$ due to the fact that the max-relative entropy for two states $\rho,\sigma \in \mathcal{S}(\mathcal{H}_A)$ can be written as
\begin{equation}\label{eq:max_relative_entropy_states}
\begin{aligned}
   e^{D_{\operatorname{max}}(\rho \Vert \sigma)} = \inf& \ \lambda& \\
    \operatorname{s.t.}& \ \rho \leq \lambda \sigma&,
\end{aligned}
\end{equation}
and the completely positive map $(\cdot) \to \rho_{A}^{1/2} (\cdot) \rho_{A}^{1/2}$ is order-preserving when $\rho_{A}$ is invertible, i.e., 
\begin{align}\label{eq:max_relativen_entropy_choi_matrices}
    \rho_{A}^{1/2}\Gamma^\mathcal{N}_{A B}\rho_{A}^{1/2} \leq \lambda \rho_{A}^{1/2}\Gamma^\mathcal{M}_{A B}\rho_{A}^{1/2} \quad \Leftrightarrow \quad \Gamma^\mathcal{N}_{A B}\leq \lambda \Gamma^\mathcal{M}_{A B}.
\end{align}
In conclusion, solving the SDP~\eqref{eq:max_relative_entropy_states} for the Choi matrices yields an adequate $\lambda$ for~\eqref{eq:jencova_representation}. This is a natural observation, because it states that the finiteness of the relative entropy of channels inherently depends on the channels themselves. Moreover, having two channels at hand and first estimating $\lambda$ allows us to determine whether  the value of the relative entropy of channels is finite. From the point of view of numerical analysis, it is desirable  to exclude this situation even before starting the algorithm.

Assuming now without loss of generality for the rest of the paper that $\lambda <+\infty$, we formulate the resulting optimization for the relative entropy of channels, together with a possible family of energy constraints $h_i(\cdot)$ on the cone $\mathcal{P}(\mathcal{H}_{A})$, as follows:
\begin{equation}
\label{eq:general_problem}
\begin{aligned}
c^\star \; \coloneqq \; \sup_{\substack{
\rho_A \ge 0,\\
\tr[\rho_A] = 1,\\
h_i(\rho_A) \ge 0
}} \;\;
& D\!\left(\rho_A^{1/2} \Gamma^\mathcal{N}_{AB} \rho_A^{1/2} \middle \Vert \rho_A^{1/2} \Gamma^\mathcal{M}_{AB} \rho_A^{1/2}\right) &
\end{aligned}
\end{equation}
under the assumption 
\begin{align}
    D_{\operatorname{max}}(\Gamma^\mathcal{N}_{A B} \Vert \Gamma^\mathcal{M}_{A B} ) = \ln  \lambda < \infty.
\end{align} 
In more detail, an energy constraint $h_i(\rho_{A})$ is defined as follows:
\begin{equation}
    h_i(\rho_{A}) \coloneqq E_i - \operatorname{tr}[H_i \rho_{A}],
\end{equation}
where $E_i \in \mathbb{R}$ and $H_i$ is a Hermitian operator playing the role of a Hamiltonian. 

\subsection{Lower bound on channel relative entropy}

\label{sec:lower-bound}

The first main result of our paper constructs a sequence of SDPs, which converges to the relative entropy of channels from below. The underlying idea is strongly related  to the developments in~\cite{koßmann2024optimisingrelativeentropysemi}, and we make use of the optimal step size for integration from~\cite[Corollary 1]{koßmann2024optimisingrelativeentropysemi}.

\begin{theorem}[Lower bound approximation for the relative entropy of channels]\label{thm:lower_bounds}
    Let $\mathcal{N}_{A\to B}$  and $\mathcal{M}_{A\to B}$ be completely positive maps with Choi matrices $\Gamma_{AB}^{\mathcal{N}}$ and $\Gamma^{\mathcal{M}}_{AB}$, respectively. Define
    \begin{equation}
    \ln\lambda \coloneqq  D_{\max}(\Gamma_{AB}^{\mathcal{N}}\Vert \Gamma_{AB}^{\mathcal{M}}).
    \end{equation}
    To find an $\varepsilon$-approximation of the relative entropy of channels from below under some semidefinite constraints $\left(h_i\right)_i$ on the optimization variable $\rho_{A}$ in~\eqref{eq:final_optimization_problem}, one has to optimize $r  = O\!\left(\sqrt{\lambda/\varepsilon}\right)$ positive semidefinite matrices $Q_k \in \mathcal{P}(\mathcal{H}_A \otimes \mathcal{H}_B)$, $1\leq k \leq r$, within the following SDP:
    \begin{equation}
\begin{aligned}
\sup_{\substack{\rho_A \ge 0,\, \tr[\rho_A] = 1 \\ 0 \le Q_k \le \rho_A \otimes \mathds{1}_B}} 
\; & \Bigg\{\tr\!\Big[\rho_A(\Gamma_{AB}^{\mathcal{N}} - \Gamma_{AB}^{\mathcal{M}})\Big] \\
& \quad + \sum_{k=1}^r \tr\!\Big[ Q_k (\alpha_k \Gamma_{AB}^{\mathcal{N}} + \beta_k \Gamma_{AB}^{\mathcal{M}}) \Big] \Bigg\}\\
& + \ln \lambda + 1 - \lambda \\
\operatorname{s.t.} \quad & h_i(\rho_A) \ge 0, \quad i = 1, \dots, m
\end{aligned}
\end{equation}
    The coefficients $\alpha_k$ and $\beta_k$ for $1\leq k \leq r$ can be efficiently precalculated by employing~\cite[Eq.~(6)]{koßmann2024optimisingrelativeentropysemi}. 
\end{theorem}

\begin{proof}
For the proof we apply the integral representation~\eqref{eq:jencova_representation} to the relative entropy expression in \eqref{eq:general_problem}.

Consider that
\begin{multline}
    \tr[\rho_A^{1/2}\Gamma_{AB}^{\mathcal{N}} \rho_A^{1/2} - \rho_A^{1/2}\Gamma_{AB}^{\mathcal{M}}\rho_A^{1/2}] \\
    = \tr[\rho_A(\Gamma_{AB}^{\mathcal{N}} - \Gamma_{AB}^{\mathcal{M}})], 
\end{multline}
which simplifies the expression for the first term in \eqref{eq:jencova_representation}, when applied to the relative entropy expression in \eqref{eq:general_problem}.

Using the integral representation in~\eqref{eq:jencova_representation}, we can rewrite~\eqref{eq:general_problem} as follows:
\begin{equation}
\begin{aligned}
\sup_{\substack{\rho_A \ge 0,\\ \tr[\rho_A] = 1}} \; 
& \Bigg\{\int_\mu^\lambda \frac{ds}{s} 
  \operatorname{tr}^+\!\Big[ \rho_A^{1/2} \Gamma_{AB}^{\mathcal{M}} \rho_A^{1/2} s 
  - \rho_A^{1/2} \Gamma_{AB}^{\mathcal{N}} \rho_A^{1/2} \Big] \\
& \qquad + \tr\!\Big[\rho_A (\Gamma_{AB}^{\mathcal{N}} - \Gamma_{AB}^{\mathcal{M}})\Big] \Bigg\} + \ln \lambda + 1 - \lambda \\
\operatorname{s.t.} \quad 
& h_i(\rho_A) \ge 0, \quad i = 1, \dots, m
\end{aligned}
\label{eq:integral-rep-sup}
\end{equation}
Writing the trace over the positive part as the supremum
\begin{multline}
    \operatorname{tr}^+\!\left[ \rho_{A}^{1/2}\Gamma^\mathcal{M}_{A B}\rho_{A}^{1/2}s - \rho_{A}^{1/2} \Gamma^\mathcal{N}_{A B}\rho_{A}^{1/2}\right] \\
    = \sup_{ 0 \leq P \leq \mathds{1}} \operatorname{tr}[P( \rho_{A}^{1/2}\Gamma^\mathcal{M}_{A B}\rho_{A}^{1/2}s - \rho_{A}^{1/2} \Gamma^\mathcal{N}_{A B}\rho_{A}^{1/2})],
    \label{eq:sup-pos-part}
\end{multline}
and interchanging the integral in~\eqref{eq:integral-rep-sup} with the supremum in~\eqref{eq:sup-pos-part} piecewise on a discretization yields a lower bound on the integral~\cite[Proposition~1]{koßmann2024optimisingrelativeentropysemi}. We choose a discretization $t_0$, $t_1$,  \ldots,  $t_r$ such that
\begin{align}\label{eq:discretization}
    \mu = t_0 < t_1 < \cdots <t_r = \lambda.
\end{align}
Then the coefficients $\alpha_k$ and $\beta_k$ are the integral values of the real-valued functions over the discretization~\eqref{eq:discretization}, which remain after exchanging the integral and supremum. These coefficients can be easily precomputed and closed formulas for them are provided in~\cite[Eq.~(6)]{koßmann2024optimisingrelativeentropysemi}. Recalling the definition of $c^{\star}$ in~\eqref{eq:general_problem}, the optimization thus becomes 
\begin{equation}
\label{eq:first_lower_bound_relaxation}
\begin{aligned}
c^\star \geq \sup_{\substack{\rho_A \ge 0,\, \tr[\rho_A] = 1 \\ 0 \le P_k \le \mathds{1}}} \; 
& \Bigg\{\sum_{k=1}^r \Big[ 
    \alpha_k \operatorname{tr}\!\big(P_k \rho_A^{1/2} \Gamma_{AB}^{\mathcal{N}} \rho_A^{1/2} \big) \\
& \quad + \beta_k \operatorname{tr}\!\big(P_k \rho_A^{1/2} \Gamma_{AB}^{\mathcal{M}} \rho_A^{1/2} \big) 
  \Big] \\
& \quad+ \tr\!\Big[\rho_A (\Gamma_{AB}^{\mathcal{N}} - \Gamma_{AB}^{\mathcal{M}}) \Big]\Bigg\} \\
& + \ln \lambda + 1 - \lambda \\
\operatorname{s.t.} \quad 
& h_i(\rho_A) \ge 0, \quad i = 1, \dots, m
\end{aligned}
\end{equation}

In a next step, we define 
\begin{align}
    Q_k \coloneqq \rho_{A}^{1/2}P_k \rho_{A}^{1/2}.   
    \label{eq:Q-k-def}
\end{align}
Importantly, we observe that conjugation with an invertible positive operator is order-preserving (here we can restrict the optimization to invertible $\rho_A$ due to the objective function being continuous in $\rho_A$). This implies the relations $Q_k\geq 0$ if and only if $P_k\geq 0$ and $Q_k\leq \rho_{A}\otimes \mathds{1}_B$ if and only if $P_k \leq \mathds{1}$. Inserting the set of variables labeled by $Q_k$ with the new positivity conditions, we can rewrite the optimization problem as
\begin{equation}
\label{eq:second_lower_bound_relaxation}
\begin{aligned}
c^\star \geq \sup_{\substack{\rho_A \ge 0,\, \tr[\rho_A] = 1 \\ 0 \le Q_k \le \rho_A \otimes \mathds{1}_B}} \; 
& \Bigg\{\sum_{k=1}^r \Big[
    \alpha_k \operatorname{tr}\!\big(Q_k \Gamma_{AB}^{\mathcal{N}}\big)  + \beta_k \operatorname{tr}\!\big(Q_k \Gamma_{AB}^{\mathcal{M}}\big)
  \Big] \\
&  + \tr\!\Big[\rho_A (\Gamma_{AB}^{\mathcal{N}} - \Gamma_{AB}^{\mathcal{M}}) \Big]\Bigg\}  + \ln \lambda + 1 - \lambda \\
\operatorname{s.t.} \quad 
& h_i(\rho_A) \ge 0, \quad i = 1, \dots, m
\end{aligned}
\end{equation}
This is the desired SDP lower bound on the relative entropy of channels. 

To get the desired $\varepsilon$-dependence, we first observe from \cite[Lem.~3]{koßmann2024optimisingrelativeentropysemi} that the lower bounds from \cite[Prop.~1]{koßmann2024optimisingrelativeentropysemi} have at least the same error dependence as the upper bounds in \cite{koßmann2024optimisingrelativeentropysemi}. Thus, the same argument to restrict to the upper bounds in the convergence analysis as in \cite{koßmann2024optimisingrelativeentropysemi} are at our disposal. Second, \cite[Cor.~3]{koßmann2024optimisingrelativeentropysemi} shows that the discretization used in \cite[Cor.~1]{koßmann2024optimisingrelativeentropysemi} yields uniform bounds in the arguments of the relative entropy. Uniformity implies immediately that \eqref{eq:second_lower_bound_relaxation} is compatible with taking a supremum and that the desired $\varepsilon$-bound is independent of the concrete arguments of the relative entropy. Thus, the discretization error can be immediately chosen from \cite[Prop.~3]{koßmann2024optimisingrelativeentropysemi}, yielding that the choice (see \cite[Cor.~1]{koßmann2024optimisingrelativeentropysemi})
\begin{equation}\label{eq:def_special_grid_xlogx}
t_k=
\begin{cases}
\mu, & k=1,\\
t_{k-1}+\sqrt{8\varepsilon\,t_{k-1}}, & k\ge 2.
\end{cases}
\end{equation}
needs $O(\sqrt{\lambda/\varepsilon})$ grid points for an $\varepsilon$-approximation of the relative entropy of channels. 
\end{proof}

\cref{thm:lower_bounds} gives a sequence of SDPs converging from below to the relative entropy of channels. The important ingredient for the SDP formulation is the integral representation~\eqref{eq:jencova_representation} with which we are able to linearize the relative entropy piecewise on the domain. In comparison to~\cite{koßmann2024optimisingrelativeentropysemi}, we need a maximization formulation, which leads to differences in the derivation after exchanging the supremum and integral. 

In all similar calculations with the integrand in the integral representation~\eqref{eq:jencova_representation}, one faces the task of estimating 
\begin{align}
    \sup_{0\leq P\leq \mathds{1}} \tr[P H]
\end{align}
for both $P$ and $H$ variables, where $H$ has to satisfy some constraints for finiteness of the expression. Usually, there are two methods that one can use to linearize this term. On one hand, one can apply the dual optimization of it, which becomes a minimization but with an operator-valued inequality constraint, as done in~\cite{koßmann2024optimisingrelativeentropysemi}. On the other hand, one can replace the product $PH$ with a new variable and adjust the constraints with an order-preserving step. As shown in \cref{thm:lower_bounds} and~\cite{koßmann2024optimisingrelativeentropysemi}, both variants can lead to a minimization or maximization formulation for the relative entropy. This makes the discretization method of the integral~\eqref{eq:jencova_representation} very powerful and applicable for many optimization tasks, which include the relative entropy. 

\subsection{Upper bound on channel relative entropy}

\label{sec:upper-bound}

In the following we aim to construct a sequence of upper bounds on~\eqref{eq:general_problem}. In minimization tasks as in~\cite{koßmann2024optimisingrelativeentropysemi}, upper bounds were constructed by using  the observation that the function
\begin{align}
    s \mapsto \tr^+[\sigma s - \rho]
\end{align}
is convex in $s$. As such, the whole proof idea worked again just with an adjustment of the integration of the coefficients. For maximization tasks, however, it is not straightforward to go along this path, because one creates a minimax formulation. In this case, we need more technical methods, as presented in the following theorem.

\begin{theorem}[Upper bound approximation for the relative entropy of channels] \label{thm:upper_bound}
Let $\mathcal{N}_{A\to B}$ and $\mathcal{M}_{A\to B}$ be two completely positive maps with Choi matrices $\Gamma_{AB}^{\mathcal{N}}$ and $\Gamma^{\mathcal{M}}_{AB}$, respectively. Define
\begin{equation}
    \ln\lambda \coloneqq  D_{\max}(\Gamma_{AB}^{\mathcal{N}}\Vert \Gamma_{AB}^{\mathcal{M}}).
\end{equation}
Suppose there exists a state $\rho_A$ such that the energy constraint with respect to a Hermitian operator $H$ on $\rho_A$ is satisfied (i.e., there exists $\rho_A$ such that $\operatorname{tr}[H \rho_A] \leq E$).
To find an $\varepsilon$-approximation of the relative entropy of channels from above, one can solve the following SDP:

\begin{equation}
\label{eq:main-result-up-bnd}
\begin{aligned}
\inf_{\substack{x \in \mathbb{R},\, y \ge 0 \\ N_0 \in \operatorname{Herm} \\ N_k \ge 0,\, 1 \le k \le r}} \; 
& x + y E + \ln \lambda + 1 - \lambda \\
\operatorname{s.t.} \quad 
& N_k \ge \gamma_k \Gamma_{AB}^{\mathcal{N}} + \delta_k \Gamma_{AB}^{\mathcal{M}}, \quad 0 \le k \le r, \\
& x \mathds{1}_A + y H_A \ge \operatorname{tr}_B\!\left[ \sum_{k=0}^r N_k \right].
\end{aligned}
\end{equation}
with coefficients $\gamma_k$ and $\delta_k$ appropriately chosen as in~\cite[Supp.~Mat.~C]{koßmann2024optimisingrelativeentropysemi}, for $k\in \{1, \ldots, r\}$, $\gamma_0 = 1$, and $\delta_0 = -1$. For the $\varepsilon$-approximation we thus need $O(\sqrt{\lambda/\varepsilon})$ SDP variables $N_0, \ldots, N_r$. 
\end{theorem}

\begin{proof}
We begin by recalling the result from~\cite[Lemma~1]{koßmann2024optimisingrelativeentropysemi}, which asserts that for quantum states $\rho, \sigma \in \mathcal{S}(\mathcal{H})$ (i.e., density matrices on a Hilbert space $\mathcal{H}$), the function defined by the trace
\begin{align}
    s \mapsto \tr^+[\sigma s - \rho] 
\end{align}
is convex and monotonically increasing. By inspecting the proof of~\cite[Lemma~1]{koßmann2024optimisingrelativeentropysemi}, it is clear that these same properties hold when $\rho$ and $\sigma$ are general positive semidefinite operators. These properties are crucial for constructing convex upper bounds in optimization problems involving quantum states (or positive semidefinite operators more generally). With this result in mind, consider the optimization problem expressed in~\eqref{eq:general_problem}. We aim to relax the problem by applying the convexity of the curve mentioned above to the function 
\begin{equation}
    s \mapsto \operatorname{tr}^+\!\left[\rho_{A}^{1/2}\Gamma^\mathcal{M}_{A B}\rho_{A}^{1/2}s - \rho_{A}^{1/2} \Gamma^\mathcal{N}_{A B}\rho_{A}^{1/2}\right].
\end{equation}
This relaxation allows us to use convex interpolation methods, which simplifies the optimization by approximating the function in question using convex combinations.

To apply the convex interpolation, we discretize the parameter space. Specifically, we introduce a set of points, $t_0$, $t_1$, \ldots ,  $t_r$, satisfying
\begin{align}\label{eq:discretization_second_proof}
    \mu = t_0 < t_1 < \cdots <t_r = \lambda,
\end{align}
which divide the range of the parameter $s$ into small intervals $[t_k, t_{k+1}]$. On each interval, we can evaluate the integral by using convex interpolation of the function over the points $t_k$ (see \cite[Eq.~(44)]{koßmann2024optimisingrelativeentropysemi} for a technical discussion on secants of convex functions). Following a similar approach as in~\cite[Supp.~Mat.~C]{koßmann2024optimisingrelativeentropysemi}, we apply convex interpolation over each of the intervals $[t_k, t_{k+1}]$ to obtain an upper bound for the optimization problem. The constants $\gamma_k$ and $\delta_k$ that appear in this interpolation are defined in~\cite[Supp.~Mat.~C]{koßmann2024optimisingrelativeentropysemi} and are related to the weights used in the convex interpolation. 
Thus, for a fixed input state~$\rho_A$, the upper bound on the original optimization problem can be written as follows, by plugging into~\eqref{eq:jencova_representation} and~\cite[Eqs.~(C2)--(C4)]{koßmann2024optimisingrelativeentropysemi}:
\begin{equation}
\label{eq:first-step-up-bnd}
\begin{aligned}
\inf_{\substack{ \nu_k \ge 0,\, 1 \le k \le r}} \; 
& \left\{ \tr\!\big[\rho_A (\Gamma_{AB}^{\mathcal{N}} - \Gamma_{AB}^{\mathcal{M}})\big] 
+ \sum_{k=1}^r \tr[\nu_k] \right\}\\
\operatorname{s.t.} \quad 
& \nu_k \ge \gamma_k \rho_A^{1/2} \Gamma_{AB}^{\mathcal{N}} \rho_A^{1/2} 
           + \delta_k \rho_A^{1/2} \Gamma_{AB}^{\mathcal{M}} \rho_A^{1/2}, \\
           &\hspace{5cm} 1 \le k \le r.
\end{aligned}
\end{equation}
where, for brevity, we have excluded the constant additive term $\ln \lambda +1-\lambda$ from the objective function for now. Applying similar reasoning as before between~\eqref{eq:Q-k-def} and~\eqref{eq:second_lower_bound_relaxation} (i.e., picking $N_k$ such that $\rho_A^{1/2}N_k\rho_A^{1/2} = \nu_k$ and observing the equivalence of various constraints), Eq.~\eqref{eq:first-step-up-bnd} can be rewritten as follows:
\begin{equation}
\begin{aligned}
\inf_{\substack{ N_k \ge 0,\\ 1 \le k \le r}} \; 
& \left\{\tr\!\big[\rho_A (\Gamma_{AB}^{\mathcal{N}} - \Gamma_{AB}^{\mathcal{M}})\big] 
+ \sum_{k=1}^r \tr\!\big[N_k (\rho_A \otimes \mathds{1}_B)\big]\right\} \\
 \operatorname{s.t.} \quad 
& N_k \ge \gamma_k \Gamma_{AB}^{\mathcal{N}} + \delta_k \Gamma_{AB}^{\mathcal{M}}, \quad 1 \le k \le r.
\end{aligned}
\end{equation}
This formulation captures the essential structure of the problem, where each $N_k$ is positive semidefinite and bounded from below by a linear combination of the Choi matrices $\Gamma_{A B}^{\mathcal{N}}$ and~$\Gamma_{A B}^{\mathcal{M}}$.

Next, we include the maximization over every quantum state $\rho_{A}$ subject to the energy constraint $\operatorname{tr}[H\rho_A]\leq E$.  This turns the problem into a minimax formulation:
\begin{equation}
\label{eq:first_minimax}
\begin{aligned}
\sup_{\substack{\rho_A \ge 0, \tr[\rho_A] = 1,\\ \tr[H\rho_A] \le E}} 
\;\inf_{\substack{N_k \ge 0 \\ 1 \le k \le r}} \;&
\Bigg\{\tr\!\big[\rho_A (\Gamma_{AB}^{\mathcal{N}} - \Gamma_{AB}^{\mathcal{M}})\big] \\
&+ \sum_{k=1}^r \tr\!\big[N_k (\rho_A \otimes \mathds{1}_B)\big]\Bigg\} \\
\operatorname{s.t.} \quad & N_k \ge \gamma_k \Gamma_{AB}^{\mathcal{N}} + \delta_k \Gamma_{AB}^{\mathcal{M}} \\ 
&\hspace{25mm} 1 \le k \le r.
\end{aligned}
\end{equation}

Now let us  observe that the term 
\begin{align}
    \tr[\rho_A(\Gamma_{AB}^{\mathcal{N}} - \Gamma_{AB}^{\mathcal{M}})]
\end{align}
can be rewritten as an optimization  with $\gamma_0 = 1$ and $\delta_0 = -1$, as follows:
\begin{equation}
\begin{aligned}
\inf_{N_0 \in \operatorname{Herm}} \;\; 
& \tr\!\big[N_0 (\rho_A \otimes \mathds{1}_B)\big] \\
\operatorname{s.t.} \quad 
& N_0 \ge \gamma_0 \Gamma_{AB}^{\mathcal{N}} + \delta_0 \Gamma_{AB}^{\mathcal{M}}.
\end{aligned}
\end{equation}
Thus, we can modify~\eqref{eq:first_minimax} as follows:
\begin{equation}
\label{eq:up-bnd-mid-step}
\begin{aligned}
\sup_{\substack{\rho_A \ge 0,\\ \tr[\rho_A] = 1,\\ \tr[H \rho_A] \le E}} 
\;\inf_{\substack{N_0 \in \operatorname{Herm} \\ N_k \ge 0,\, 1 \le k \le r}} \;&
\sum_{k=0}^r \tr\!\big[N_k (\rho_A \otimes \mathds{1}_B)\big] \\
\operatorname{s.t.} \quad & N_k \ge \gamma_k \Gamma_{AB}^{\mathcal{N}} + \delta_k \Gamma_{AB}^{\mathcal{M}},\\
&\hspace{18mm}\quad 0 \le k \le r.
\end{aligned}
\end{equation}
The fact that $\rho_{A}$ comes from the spectrahedron $\left\{\rho_A \in \mathcal{P}(\mathcal{H}_A) \ \vert \ \tr[\rho_A] = 1\right\}$ yields convexity and compactness of the domain for $\rho_A$. In addition, linearity of the objective function and convexity of the set of operators $N_0, \ldots, N_r$ satisfying the constraints enables us to apply Sion's minimax theorem~\cite{sion1958GeneralMinimaxTheorems}. Thus, Eq.~\eqref{eq:up-bnd-mid-step} is equal to
\begin{equation}\label{eq:one_step_before_infsup}
\begin{aligned}
\inf_{\substack{N_0 \in \operatorname{Herm} \\ N_k \ge 0,\, 1 \le k \le r}} 
\;\sup_{\substack{\rho_A \ge 0,\\ \tr[\rho_A] = 1,\\ \tr[H \rho_A] \le E}} \;&
\sum_{k=0}^r \tr\!\big[N_k (\rho_A \otimes \mathds{1}_B)\big] \\
\operatorname{s.t.} \quad & N_k \ge \gamma_k \Gamma_{AB}^{\mathcal{N}} + \delta_k \Gamma_{AB}^{\mathcal{M}} \\
&\hspace{18mm}0 \le k \le r.
\end{aligned}
\end{equation}

Now we introduce the operator 
\begin{align}
    J_{AB} \coloneqq \sum_{k = 0}^r N_k
\end{align}
and just consider the program involved in  the supremum
\begin{equation}\label{eq:sup_formulation_upper_bounds}
\begin{aligned}
\sup_{\substack{\rho_A \ge 0,\\ \tr[\rho_A] = 1,\\ \tr[H \rho_A] \le E}} \;&
\tr\!\big[J_{AB} (\rho_A \otimes \mathds{1}_B)\big].
\end{aligned}
\end{equation}

An explicit calculation of the dual program of~\eqref{eq:sup_formulation_upper_bounds} is given in \cref{lem:dual_program},
and in particular, we give an argument for strong duality under the assumption that there exists a state $\rho_A$ satisfying the energy constraint (i.e., there exists $\rho_A$ satisfying $\tr[H \rho_{A} ] \leq E$). Thus we can replace~\eqref{eq:sup_formulation_upper_bounds} with the following program, which has  the same value as~\eqref{eq:sup_formulation_upper_bounds}, due to strong duality:
\begin{equation}\label{eq:inf_after_lemma}
\begin{aligned}
\inf_{\substack{x \in \mathbb{R}, y \ge 0}} \;&\  x + y E \\
\operatorname{s.t.} \;& x \mathds{1}_A + y H_A \ge \operatorname{tr}_B[J_{AB}].
\end{aligned}
\end{equation}

Replacing now~\eqref{eq:sup_formulation_upper_bounds} in~\eqref{eq:one_step_before_infsup} with~\eqref{eq:inf_after_lemma} and also replacing $ J_{AB}$ with $\sum_{k=0}^r N_k $ yields the following program:
\begin{equation}
\begin{aligned}
\inf_{\substack{x \in \mathbb{R}, y \ge 0,\\ N_0 \in \operatorname{Herm},\\ N_k \ge 0,\, 1 \le k \le r}} \;& 
x + y E \\
\operatorname{s.t.} \;& N_k \ge \gamma_k \Gamma_{AB}^{\mathcal{N}} + \delta_k \Gamma_{AB}^{\mathcal{M}}, \quad 0 \le k \le r,\\
& x \mathds{1}_A + y H_A \ge \operatorname{tr}_B \!\left[ \sum_{k=0}^r N_k \right].
\end{aligned}
\end{equation}
After incorporating the additive constant $\ln \lambda + 1-\lambda$ back into the objective function, this concludes the proof of~\eqref{eq:main-result-up-bnd}.

To get the desired $\varepsilon$-dependence, we observe that the upper bounds in \cite[Prop. 3]{koßmann2024optimisingrelativeentropysemi} are uniform in the states. Thus, the $\varepsilon$-dependence does not depend on the arguments of the relative entropy, which yields that we can choose the grid from \eqref{eq:def_special_grid_xlogx} and conclude with the same uniformity argument as in \cite[Cor. 3]{koßmann2024optimisingrelativeentropysemi} the desired dependence $O(\sqrt{\lambda/\varepsilon})$ for the number of grid points. 
\end{proof}

\cref{thm:upper_bound} provides a dual characterization of the optimization problem by applying \cref{lem:dual_program}. This application results in a final expression that is a combination of linear programming and semidefinite programming variables.

The combination of multiple technical tools from~\cite{koßmann2024optimisingrelativeentropysemi} and~\cref{lem:dual_program} is key in obtaining a minimization formulation for the relative entropy  of completely positive maps. This is significant because, in general, maximizing the relative entropy is a highly non-trivial task. The primary difficulty arises from the inherent non-linearity, %and convexity of the problem
making it challenging to analyze and optimize directly.

\subsection{Discussion}

Both \cref{thm:lower_bounds} and \cref{thm:upper_bound} reveal that the relative entropy of quantum channels has a particular structure that simplifies the problem in certain respects. This is due to the local nature of the quantum state $\rho_A \in \mathcal{S}(\mathcal{H}_A)$, which plays a critical role in the problem's formulation. On the one hand, this allows us to use techniques that upper bound the SDP variables $Q_k$ with respect to $\rho_{A}$, simplifying the analysis. On the other hand, we can apply the dualization techniques from~\cref{lem:dual_program} to further refine the bounds and the characterization of the problem.

Overall, the combination of these techniques leads to a more tractable formulation for the relative entropy of channels, particularly in the presence of local quantum operations, making the optimization problem more approachable than it might initially appear.

\section{Applications in resource theories of channels}

\label{sec:application_resource_theory}

An important application of the theorems presented in \cref{sec:main_results} lies within  resource theories of quantum channels~\cite{CFS16,TEZP19,liu2019resourcetheoriesquantumchannels,LY20}. Resource theories have become an essential part of quantum information theory~\cite{Chitambar_2019}, providing a systematic way for quantifying and understanding quantum resources. The importance of resource theories stems from their flexibility in adapting to various operational tasks, allowing us to tailor the theory to a specific problem. 

The extension of resource theories from quantum states to  channels has significantly broadened the scope of these theories. While resource theories for states focus on static resources --- essentially the properties of quantum states --- resource theories for channels address dynamical resources, which are associated with quantum processes or transformations. A key example of this is that resource theories of channels encompass resource theories of states by treating preparation channels (those that prepare states) as a special case.

A central concept in resource theories is the notion of a resource monotone, which is a measure that quantifies the amount of resource present and that does not increase under free operations. In the context of quantum channels, one can construct resource monotones built from the relative entropy of channels. Importantly, it also frequently provides a bound on optimal rates of various operational tasks involving the use of channels. For instance, it can be used to determine how well certain channels can perform specific tasks, or to compare the effectiveness of different channels in resource consumption. The results in \cref{sec:main_results} provide a practical way to compute the relative entropy of channels, which is vital for applying resource theories to real-world problems. By being able to calculate these quantities, we gain the ability to establish concrete bounds on the resources involved in a given application. This has significant implications for optimizing quantum operations and understanding the limits of quantum processes.

\subsection{Background on resource theories of channels}

\label{subsec:recap_resource_theory_channels}

To set the stage for these applications, let us briefly review resource theories of channels, as introduced in~\cite{Gour_2021,Cooney_2016,liu2019resourcetheoriesquantumchannels}. This review focuses on the aspects necessary to understand how the relative entropy of channels plays a role in practical applications within this framework.

Consider two quantum systems associated with $\mathcal{H}_A$ and $\mathcal{H}_B$, and define the set 
\begin{align}
    \operatorname{CPTP}(A\to B) \coloneqq \{T\colon \mathcal{S}(\mathcal{H}_A) \to \mathcal{S}(\mathcal{H}_B) \ \vert \ T \ \text{channel}\}.
\end{align}
In resource theories we characterize a subset $\mathcal{F}(A\to B) \subseteq \operatorname{CPTP}(A\to B)$ as the set of free channels. Free means operationally that using channels in $\mathcal{F}(A\to B)$ comes at no cost. Furthermore we enforce the following rules for a set of free operations:
\begin{enumerate}
    \item $\mathcal{F}$ is closed under tensor products and composition.
    \item $\mathcal{F}(A\to B)$ is topologically closed.
    \item $\mathcal{F}(A\to A)$ contains the identity. 
\end{enumerate}
For particular cases, one may define some more requirements for the set of free objects, but it is already sufficient to define $\mathcal{F}$ with these three axioms for our general purpose of introducing the fundamental framework. 

For our applications, we consider resource monotones next. Usually, questions in resource theories ask for the cost of a certain task under the assumption that a certain set $\mathcal{F}(A\to B)$ can be assumed to be free. As an example, one can seek to determine the cost of simulating a quantum channel. To answer this question reasonably, we need to define a measure for channels not included in $\mathcal{F}(A\to B)$. Axiomatically, we follow~\cite{liu2019resourcetheoriesquantumchannels}, by asserting that a resource monotone $\Omega\colon  \operatorname{CPTP}(A\to B) \to \mathbb{R}$ satisfies the following axioms:
\begin{enumerate}
    \item Normalization: $\Omega(\mathcal{N}) = 0$ if $\mathcal{N} \in \mathcal{F}(A\to B)$ and $\Omega(\mathcal{N})\geq 0$ for all channels.
    \item Faithfulness: $\Omega(\mathcal{N}) = 0$ if and only if $\mathcal{N}\in \mathcal{F}(A\to B)$ and $\Omega(\mathcal{N}) >0$ otherwise.
    \item Monotonicity under composition with a free channel $\mathcal{M}\in \mathcal{F}(A\to B)$
    \begin{itemize}
        \item left composition: $\Omega(\mathcal{M}\circ \mathcal{N}) \leq \Omega(\mathcal{N})$
        \item right composition: $\Omega(\mathcal{N}\circ \mathcal{M}) \leq \Omega(\mathcal{N})$
        \item tensoring: $\Omega(\mathcal{M}\otimes \mathcal{N}) \leq \Omega(\mathcal{N})$
    \end{itemize}
    \item Convexity: $\sum_i p_i\Omega(\mathcal{N}_i) \geq \Omega\!\left(\sum_i p_i \mathcal{N}_i\right)$ for a probability distribution $\left(p_i\right)_i$ over channels $\left(\mathcal{N}_i\right)_i$. 
\end{enumerate}
The axioms for resource monotones have a deep operational meaning. The requirements are in a way such that measuring a resource can be identified with a ``cost'' of applying, e.g., the considered channel. 

As shown in~\cite{liu2019resourcetheoriesquantumchannels}, the global robustness of a channel $\mathcal{N}$ is defined as 
\begin{align}
    R(\mathcal{N}) \coloneqq \min \!\left\{s\geq 0 \ \middle \vert \ \frac{1}{s+1}\mathcal{N} + \frac{s}{1+s}\mathcal{N}^\prime   \in \mathcal{F}\right \}.
\end{align}
and is a resource monotone. Furthermore we can build the log-robustness out of it:
\begin{equation}
\begin{aligned}
    \operatorname{LR}(\mathcal{N}) &\coloneqq \ln  (1+R(\mathcal{N})) \\
    &= \min_{\mathcal{M} \in \mathcal{F}} D_{\operatorname{max}}(\mathcal{N}\Vert \mathcal{M}).
\end{aligned}
\end{equation}
Now it is easy to show that the log-robustness is a resource monotone~\cite[Prop.~8]{liu2019resourcetheoriesquantumchannels}, and in particular, the smoothed log-robustness can be considered, if one optimizes over the $\varepsilon$-ball in diamond norm. The difference with resource theories of states, and thus the reason why we have no asymptotic equipartition theorem~\cite{Tomamichel_2009}, is that the dual to the diamond norm is the trace distance and not the purified distance. Thus we expect in general just lower bounds. This was done in~\cite[Theorem 11]{liu2019resourcetheoriesquantumchannels}, where it was shown that, for all $\varepsilon\in(0,1)$,
\begin{align}
    \liminf_{n\to \infty} \frac{1}{n} D_{\operatorname{max}}^\varepsilon(\mathcal{N}^{\otimes n} \Vert \mathcal{M}^{\otimes n}) \geq D(\mathcal{N}\Vert \mathcal{M})
\end{align}
and the right-hand side is the relative entropy of channels. Therefore we conclude that estimates on the relative entropy of channels give estimates on the asymptotic rates for channel simulation tasks.  

\subsection{Optimization task}\label{subsec:optimization_task}

From now on we assume to have a resource theory of channels at hand and characterized, as described in \cref{subsec:recap_resource_theory_channels}, by a set of free channels, $\mathcal{F}$. Then, as the relative entropy of channels can be seen as a measure of distinguishability, we aim to solve the following optimization task for a fixed channel $\mathcal{N}_{A\to B}$:
\begin{equation}
    \begin{aligned}
        \inf \ &D(\mathcal{N}_{A\to B}\Vert \mathcal{M}_{A\to B})& \\
          \operatorname{s.t.} \  &\mathcal{M} \in \mathcal{F}.&
    \end{aligned}
    \label{eq:rel-ent-ch-res-th}
\end{equation}
With the observation in~\eqref{eq:final_optimization_problem}, we can rewrite the relative entropy of channels as an optimization over a quantum state $\rho_A \in \mathcal{S}(\mathcal{H}_A)$ and thus we may introduce an additional energy constraint on $\rho_A$.  
With these preliminaries established, we can just apply our developed methods in \cref{thm:upper_bound} to obtain the following optimization task, which provides an upper bound on \eqref{eq:rel-ent-ch-res-th}:
\begin{equation}
\label{eq:channel_optimization_resource_theories}
\begin{aligned}
&\inf_{\substack{
x \in \mathbb{R},\; y \ge 0,\\
N_0 \in \operatorname{Herm},\; N_k \ge 0\; (1\le k \le r),\\
\Gamma_{AB}^{\mathcal{M}} \ge 0,\; 
\operatorname{tr}_B[\Gamma_{AB}^{\mathcal{M}}] = \mathds{1}_A,\; 
\Gamma_{AB}^{\mathcal{M}} \in \mathcal{F}
}} \;
\left\{x + y E\right\} + \ln \lambda + 1 - \lambda \\
&\qquad \operatorname{s.t.} \; N_k \ge \gamma_k \Gamma_{AB}^{\mathcal{N}} + \delta_k \Gamma_{AB}^{\mathcal{M}}, \quad 0 \le k \le r &\\
& \qquad x \mathds{1}_A + y H_A \ge \operatorname{tr}_B \!\left[ \sum_{k=0}^r N_k \right] .
\end{aligned}
\end{equation}
The constraint $\Gamma_{AB}^{\mathcal{M}} \geq 0$ enforces $\mathcal{M}$ to be completely positive, the constraint $\operatorname{tr}_B[\Gamma_{AB}^{\mathcal{M}}] = \mathds{1}_A$ enforces $\mathcal{M}$ to be trace preserving, and the constraint $\Gamma_{AB}^{\mathcal{M}} \in \mathcal{F}$ is a slight abuse of notation, indicating that $\mathcal{M}$ should be a free channel.

As it turns out, Eq.~\eqref{eq:channel_optimization_resource_theories} naturally gives a formulation in terms of the set $\mathcal{F}$. Thus, the hardness of evaluating~\eqref{eq:channel_optimization_resource_theories} now reduces to the question of whether $\mathcal{F}$ is SDP representable. This is of course a general problem, where possibilities from very simple up to impossible can occur. Two possible choices for $\mathcal{F}$ are  positive partial transpose (PPT) channels~\cite{Rai99,Rai01} or $k$-extendible channels~\cite{PBHS11,KDWW19,KDWW21,Berta_2021}. These channels are specified by semidefinite constraints, and thus the task to estimate the quantum relative entropy of a given channel to either of these sets of channels reduces to an SDP. On the other hand, if we aim to optimize over all entanglement-breaking channels~\cite{Horodecki_2003} or separable channels~\cite{VPRK97,BNS98}, then we would need to optimize over a restricted set of separable states~\cite[Thm.~4.6.1]{wilde2017QuantumInformationTheory} on $\mathcal{H}_A \otimes \mathcal{H}_B$, which is not SDP representable~\cite{Fawzi2021} and an inherently hard task.

Another concrete example for an SDP with a certain set $\mathcal{F}$ in the spirit of~\eqref{eq:channel_optimization_resource_theories} is given by the energy-constrained entanglement-assisted capacity of a channel~\cite{Holevo03}. Here $\mathcal{F}$ is the set of replacer channels (i.e., channels of the form $\mathcal{R}(\omega) = \operatorname{tr}[\omega]\sigma$) and thus can be directly denoted by the state $\sigma_B \in \mathcal{S}(\mathcal{H}_B)$, which is the state prepared by such a replacer channel:
\begin{equation}
\begin{aligned}
\inf_{\substack{
x \in \mathbb{R},\; y \ge 0,\\
N_0 \in \operatorname{Herm},\; N_k \ge 0\; (1\le k \le r),\\
\sigma_B \ge 0,\; \operatorname{tr}[\sigma_B] = 1
}} \;& x + y E + \ln \lambda + 1 - \lambda &\\
\operatorname{s.t.} \;& N_k \ge \gamma_k \Gamma_{AB}^{\mathcal{N}} + \delta_k \mathds{1}_A \otimes \sigma_B, \\
&\hspace{3cm} 0 \le k \le r &\\
& x \mathds{1}_A + y H_A \ge \operatorname{tr}_B \!\left[ \sum_{k=0}^r N_k \right] . &
\end{aligned}
\end{equation}

Let us note here that prior work employed a Blahut--Arimoto-like algorithm to estimate the unconstrained entanglement-assisted capacity~\cite{RISB20,RISB21}, but our approach here allows for estimating it via an SDP. We leave a detailed comparison of these approaches for future work.

\section{Examples}\label{sec:examples}

\subsection{\texorpdfstring{Calculating relative entropy of channels for fixed channels $\mathcal{N}$ and $\mathcal{M}$}{Calculating relative entropy of channels for fixed channels N and M}}\label{subsec:numerical_evidence}

In this section, we provide a proof-of-principle example for the relative entropy of channels for qubit channels. Let us consider the dephasing channel $\mathcal{N}_{\operatorname{deph}}$ with fixed parameter $0.4$ and the depolarizing channel $\mathcal{M}_{\operatorname{dep}}$ with variable parameter $p\in[0,0.1]$.
We abbreviate the Pauli matrices by $\sigma_x,\sigma_y,\sigma_z$. The dephasing channel $\mathcal{N}_{\operatorname{deph}}$ with parameter $p_{\mathrm{deph}}=0.4$ is defined by
\begin{align}
    \mathcal{N}_{\operatorname{deph}}(\rho)
= p_{\mathrm{deph}}\,\rho + (1-p_{\mathrm{deph}})\,\sigma_z \rho \sigma_z
\end{align}
and the depolarizing channel as
\begin{align}
    \mathcal{M}_{\operatorname{dep}}(\rho)
=\Bigl(1-\tfrac{3p}{4}\Bigr)\rho
+ \tfrac{p}{4}\Bigl(\sigma_x \rho \sigma_x + \sigma_y \rho \sigma_y + \sigma_z \rho \sigma_z\Bigr).
\end{align}
For the precision we choose $\varepsilon = 10^{-2}$. 

\autoref{fig:plot_dephsing_vs_depolarizing} depicts the relative entropy of these two channels as a function of the depolarizing parameter $p$, while the dephasing parameter is fixed. As expected, the channels become less distinguishable as the depolarizing parameter increases. A github implementation of \autoref{fig:plot_dephsing_vs_depolarizing} is provided in \cite{numerics}. To provide numerical justification for the results of \autoref{thm:lower_bounds} and \autoref{thm:upper_bound}, we examine the qualitative dependence of the required number of grid points on fixed values of $\lambda$ and $\varepsilon$, as illustrated in \autoref{fig:data_evaluation}. In the implementation, we first solve the SDP in \eqref{eq:max_relative_entropy_states}, applied as shown in \eqref{eq:max_relativen_entropy_choi_matrices}, to compute $\lambda$. We assume a numerical accuracy of at least $\varepsilon = 10^{-2}$. Our main claim in \autoref{thm:lower_bounds} and \autoref{thm:upper_bound} is that the number of grid points $r \in \mathbb{N}$ required to achieve an $\varepsilon$-approximation for a given $\lambda \in \mathbb{R}_{\geq 0}$ satisfies
\begin{align}
  r = O\!\left(\sqrt{\tfrac{\lambda}{\varepsilon}}\right).
\end{align}
The plot confirms that this estimate is satisfied in the considered cases. For completeness, we include in \autoref{fig:error_dependencies} the gap error resulting from \autoref{thm:lower_bounds} and \autoref{thm:upper_bound}, which is consistently of order $\varepsilon = 10^{-2}$.

\begin{figure}
\begin{tikzpicture}
    \begin{axis}[scaled x ticks=false,
xticklabel=\pgfkeys{/pgf/number format/.cd,fixed,precision=2,zerofill}\pgfmathprintnumber{\tick},
        width=8cm, % Adjust the width
        height=8cm, % Adjust the height
        xlabel={depolarizing $p $}, % X-axis label
        ylabel={$D(\mathcal{N}_{\operatorname{deph}} \| \mathcal{M}_{\operatorname{dep}})$}, % Y-axis label
        grid=major, % Add gridlines
        legend style={at={(0.95,0.95)},anchor=north east},
        title={relative entropy of Channels} % Optional title
    ]
    \addplot[
        color=PineGreen,
        mark=*,
        thick,
        mark options={fill=PineGreen, draw=PineGreen}
    ] table [col sep=space] {qubit_values.txt};
    %\legend{Data Points}
    \end{axis}
\end{tikzpicture}
    \caption{Relative entropy of channels $D(\mathcal{N}_{\operatorname{deph}} \Vert \mathcal{M}_{\operatorname{dep}})$ between the dephasing channel $\mathcal{N}_{\operatorname{deph}}$ (with fixed parameter $0.4$) and the depolarizing channel $\mathcal{M}_{\operatorname{dep}}$ as the depolarizing parameter $p$ varies in $[0,0.1]$. The solid line represents the mean value of upper and lower bound (from \autoref{thm:lower_bounds} and \autoref{thm:upper_bound}) calculated in \cite{numerics}. The precision of these bounds is better than $\varepsilon = 10^{-2}$.}
    
    \label{fig:plot_dephsing_vs_depolarizing}
\end{figure}

\begin{figure}
    \centering
    \begin{tikzpicture}
\begin{axis}[
  width=8cm, height=8cm,
  xmode=log, ymode=log,
  xmin=200, xmax=1200,
  xlabel={$\sqrt{\lambda/\varepsilon}$},
  ylabel={$r$ (grid points)},
  grid=major,
  legend pos= north west 
]
  % Data, sorted by x
   \addplot+[
    only marks,
    mark=*,
    mark size=1.8pt,
    mark options={fill=PineGreen, draw=PineGreen}
  ] table[
    x=x, y=r, sort
  ] {\datatable};
  \addlegendentry{$(\sqrt{\lambda/\varepsilon},r)$}

  % Straight guide line with gradient 0.845: r = 0.845 x
   \addplot[domain=200:1200, thick, color=MidnightBlue] {0.845*x};
  \addlegendentry{$r=0.845\,x$}
\end{axis}
\end{tikzpicture}
    \caption{We validate the claimed complexity bound $O\big(\sqrt{\lambda/\varepsilon}\big)$ for the number of grid points. For this purpose we plot the calculated number of grid points for the experiment in \autoref{fig:plot_dephsing_vs_depolarizing} against the calculated $\exp(D_{\operatorname{max}}(\mathcal{N}_{\operatorname{deph}} \Vert \mathcal{M}_{\operatorname{dep}}))$ divided by the calculated gap value in each experiment and taking the square root. As we see, the quotients are all below a straight line with gradient $0.845$ for the two channels under consideration. The resulting SDP always has as many variables of size $d^2 \times d^2$ as there are grid points, where $d$ is the dimension of the system on which the channel acts.}
    \label{fig:data_evaluation}
\end{figure}

\begin{figure}
    \centering
\begin{tikzpicture}
\begin{axis}[scaled x ticks=false,
xticklabel=\pgfkeys{/pgf/number format/.cd,fixed,precision=2,zerofill}\pgfmathprintnumber{\tick},
  width=8cm, height=6cm,
  xlabel={depolarizing $p $}, % X-axis label,
  ylabel={$-\log_{10}(\varepsilon)$},
  grid=major
]
  \addplot+[only marks, mark=*, mark size=1.5pt,mark options={fill=PineGreen, draw=PineGreen}]
    table[x=x, y=y] {\datatableB};
\end{axis}
\end{tikzpicture}
    \caption{We plot the gap value for the channels considered in \autoref{fig:plot_dephsing_vs_depolarizing}. Comparing this with \autoref{fig:data_evaluation} yields that we at least match the error dependencies claimed in \autoref{thm:lower_bounds} and \autoref{thm:upper_bound}.}
    \label{fig:error_dependencies}
\end{figure}

\subsection{\texorpdfstring{Optimizing $\mathcal{M}$ over the set of $k$-extendible channels}{Optimizing M over the set of k-extendible channels}}\label{subsec:k-extendible}

A particularly interesting class of channels is the set of entanglement breaking channels. Entanglement breaking means that, on whatever input state $\rho_{AB}$ we apply the channel $\mathcal{M}_{A\to A^\prime}$, the output state $\mathcal{M}_{A\to B}(\rho_{AB^\prime}) \in \operatorname{SEP}(B:B^\prime)$ is always separable. As shown in \cite[Thm. 4]{Horodecki_2003}, a channel is entanglement breaking if and only if its Choi matrix $\Gamma_{BB^\prime}^\mathcal{M}$ is separable. As already discussed, the set of separable states is not SDP representable~\cite{Fawzi2021} and thus it is an inherently hard task to optimize over it. However, for solving this task, quantum de Finetti theorems have been developed.

We review through the necessary arguments for the following task. Given an arbitrary channel $\mathcal{N}_{A\to B}$ we aim to calculate
\begin{equation}\label{eq:opt_entanglement_breaking}
\begin{aligned}
    \inf \ &D(\mathcal{N}\| \mathcal{M})& \\
    \operatorname{s.th.} \ &\mathcal{M} \ \text{is entanglement breaking}.&
\end{aligned}
\end{equation}
With the equivalence of entanglement-breaking channels and separable Choi matrices, we conclude that the optimization problem \eqref{eq:opt_entanglement_breaking} is equivalent to 
\begin{equation}\label{eq:optimization_problem_kextendible}
\begin{aligned}
\inf_{\substack{
\Gamma_{AB}^{\mathcal{M}} \in \operatorname{SEP}(A:B),\\
\operatorname{tr}_B[\Gamma_{AB}^{\mathcal{M}}] = \mathds{1}_A
}} \;\; 
\sup_{\rho_A \in \mathcal{S}(\mathcal{H}_A)} \;&
D\!\left(\rho_A^{1/2} \Gamma_{AB}^{\mathcal{N}} \rho_A^{1/2} \middle\| \rho_A^{1/2} \Gamma_{AB}^{\mathcal{M}} \rho_A^{1/2}\right) .&
\end{aligned}
\end{equation}
A precise analysis of the problem above yields that we need a de Finetti theorem without side-information but with constraints. Recent work done in \cite{Kossmann2025} (which is adopted from \cite{Berta_2021}) provides an in-depth analysis of the different types of SDP hierarchies and provides in particular a proposition for our situation (Prop. $3.2$ therein). Similar to \cite{Kossmann2025} we define the following sets
\begin{align}
    \Sigma_{\operatorname{Ext}}^n(A:B) \coloneqq \left\{\rho_{AB} \in \operatorname{Ext}^n(A:B) \ \middle \vert \ \tr_B[\rho_{AB}] = \frac{\mathds{1}_A}{d_A} \right\}.
\end{align}
whereby $\operatorname{Ext}^n(A:B)$ defines the set of states which have an $n$-extension (see, e.g., \cite[Eqs.~(13)--(14)]{SW2025}). 
In particular we define
\begin{align}
    \Sigma_{\operatorname{SEP}}(A:B) \coloneqq \left\{\rho_{AB} \in \operatorname{SEP}(A:B) \ \middle \vert \ \tr_B[\rho_{AB}] = \frac{\mathds{1}_A}{d_A} \right\}.
\end{align}
Then we have the natural sequence of inclusions
\begin{align}
    \Sigma_{\operatorname{Ext}}^1(A:B) \supseteq \cdots \supseteq \Sigma_{\operatorname{Ext}}^n(A:B) \supseteq  \cdots \supseteq  \Sigma_{\operatorname{SEP}}(A:B).
\end{align}
Now, the quantum de Finetti theorem with constraints in \cite[Prop. $3.2$]{Kossmann2025} shows that the sequence of sets converges towards $\Sigma_{\operatorname{SEP}}(A:B)$. Furthermore, the following expression considered as a function of $\Gamma_{AB}^\mathcal{M}$
\begin{align}
    \Gamma_{AB}^\mathcal{M} \mapsto \sup_{\rho_A} D\!\left(\rho_A^{1/2}\Gamma_{AB}^\mathcal{N} \rho_A^{1/2} \middle\| \rho_A^{1/2}\Gamma_{AB}^\mathcal{M} \rho_A^{1/2} \right) 
\end{align}
is lower semicontinuous. Hence, standard arguments from convex analysis yield
\begin{multline}
    \lim_{n\to \infty} \inf_{\Gamma_{AB}^\mathcal{M} \in \Sigma_{\operatorname{Ext}}^n(A:B)} D\!\left(\rho_A^{1/2}\Gamma_{AB}^\mathcal{N} \rho_A^{1/2} \middle\| \rho_A^{1/2}\Gamma_{AB}^\mathcal{M} \rho_A^{1/2} \right)\\ = \inf_{\Gamma_{AB}^\mathcal{M} \in \Sigma_{\operatorname{SEP}}(A:B)} D\!\left(\rho_A^{1/2}\Gamma_{AB}^\mathcal{N} \rho_A^{1/2} \middle\| \rho_A^{1/2}\Gamma_{AB}^\mathcal{M} \rho_A^{1/2} \right) .
\end{multline}
Now we are in a position to combine those findings above with our result. Thus, in order to solve \eqref{eq:optimization_problem_kextendible}, we replace the set of separable Choi matrices  with the set $\Sigma_{\operatorname{Ext}}^n(A:B)$ and get a lower bound on the relative entropy of channels for the resource theory of non-entanglement-breaking channels
\begin{equation}\label{eq:optimization_problem_kextendible_replaced}
\begin{aligned}
    \inf \sup_{\rho_A \in\mathcal{S}(\mathcal{H}_A)}\ &D\!\left(\rho_A^{1/2}\Gamma^\mathcal{N}_{AB} \rho_A^{1/2}\middle\| \rho_A^{1/2}\Gamma_{AB}^\mathcal{M}\rho_A^{1/2}\right)& \\
    \operatorname{s.th.} \ &\Gamma_{AB}^\mathcal{M} \in \Sigma^n_{\operatorname{SEP}}(A:B).& 
\end{aligned}
\end{equation}
The set $\Sigma^n_{\operatorname{SEP}}(A:B)$ is SDP-representable and thus we can calculate with \autoref{thm:upper_bound} an $\varepsilon$-approximation for \eqref{eq:optimization_problem_kextendible_replaced} with the help of the following program
\begin{equation}\label{eq:channel_optimization_resource_theories_kext}
\begin{aligned}
\inf_{\substack{
x \in \mathbb{R}, \; y \ge 0,\\
N_0 \in \operatorname{Herm}, \; N_k \ge 0, \; 1 \le k \le r,\\
\Gamma_{AB}^{\mathcal{M}} \in \Sigma_{\operatorname{Ext}}^n(A:B)
}} \;\;
& x + y E + \ln \lambda + 1 - \lambda &\\
\operatorname{s.t.} \;& N_k \ge \gamma_k \Gamma_{AB}^{\mathcal{N}} + \delta_k \Gamma_{AB}^{\mathcal{M}}, \quad 0 \le k \le r &\\
& x \mathds{1}_A + y H_A \ge \operatorname{tr}_B \!\left[ \sum_{k=0}^r N_k \right] .&
\end{aligned}
\end{equation}

Here we focused on channels that are point-to-point (i.e., from a sender to a receiver), but note that the development above can be generalized to bipartite channels and the resource theory of unextendibility put forward in \cite{KDWW19,KDWW21}.

For numerics two comments may be of interest. First of all, we could include directly just a PPT constraint (this would give a different resource theory). Second, as discussed in \cite{Kossmann2025}, $n$-extendible states with constraints but without side-information can be represented within the symmetric subspace such that we can parameterize $\Sigma_{\operatorname{Ext}}^n(A:B)$ with significantly less SDP variables. However, to enforce the constraints we need one calculation over an exponentially large sum.

\section{Conclusion}

\label{sec:conclusion}

In this work, we provided a practical method for calculating the relative entropy of quantum channels. As established  earlier~\cite{liu2019resourcetheoriesquantumchannels}, the relative entropy of channels holds a fundamental position in resource theories of channels, analogous to the role of the standard (Umegaki) relative entropy in resource theories for quantum states. The approach we took builds on a recent advance~\cite{koßmann2024optimisingrelativeentropysemi}, which proposed optimizing the relative entropy of quantum states using an integral representation derived  in~\cite{jenčová2023recoverabilityquantumchannelshypothesis,Frenkel_2023}.

A key insight of our work is the realization that the discretized linearization of this integral representation is essential for handling the maximization tasks inherent in working with the convex and nonlinear nature of relative entropy. This approach allowed us to contribute technically by extending the minimization techniques from~\cite{koßmann2024optimisingrelativeentropysemi} to the maximization problems that arise in our study of quantum channels. It is important to note that, while minimization tasks involving relative entropy can often be efficiently solved using interior-point methods, such as those developed in~\cite{fawzi2022semidefiniteprogramminglowerbounds,fawzi2023optimalselfconcordantbarriersquantum}, the maximization tasks present unique challenges. Unlike minimization problems, which can be handled in a systematic or ``black-box'' manner, maximizing the relative entropy is more complex. Exceptions to this are semidefinite programs, where the objective function is linear, making it both convex and concave, and thus amenable to established optimization techniques. Even though the quantum relative entropy of channels can be seen as a concave optimization problem in the state $\rho_A \in \mathcal{S}(\mathcal{H}_A)$ (\cite[Prop.~13]{Wang2019magic} and \cite[Prop.~7.83]{khatriwilde2024}) in~\eqref{eq:final_optimization_problem}, finding a conic optimization representation including the relative entropy is a challenging task. 
%Thus, clearly on one hand the hope for applicable techniques is justified due to concavity, and on the other hand appropriate tools might be difficult to find in order to get approximation due to the nonlinear and convex functional expression. 
As we have shown, in this situation an integral representation which solves the second challenge in a natural way can be very advantageous. 

Beyond the technical advancements in this paper, there is significant potential for future exploration in the area of maximizing relative entropy. For instance, the problem of calculating the measured relative entropy had been unresolved for a long time due to the lack of a convex or concave variational expression, as described in~\cite{Berta_2017}. However, recent progress using semidefinite programming techniques and beyond
has led to a solution \cite{huang2024semidefiniteoptimizationmeasuredrelative,huang2025acceleratedoptimizationmeasuredrelative}. 
It would be intriguing to investigate whether our approach could also be applied to measured relative entropy. Furthermore, uncertainty relations involving Shannon-type entropies can often be framed as maximization problems involving relative entropy, yet relatively few methods are available to address these challenges effectively for provable bounds. It would be interesting to explore whether our techniques could offer new solutions in this area. We also find it interesting to determine whether the relative entropy of channels could be computed by a Blahut--Arimoto-like algorithm, along the lines of~\cite{RISB20,RISB21}.

Lastly, let us note that the methods we have developed are not limited to the specific case of relative entropy. They are also applicable to more general $f$-divergences, as defined in~\cite{Hirche_2024}. There is potential to extend the framework in~\cite{koßmann2024optimisingrelativeentropysemi} to encompass general $f$-divergences, opening up new avenues for research.  We leave it as an open contribution for future research and hope to see it together with a well motivated operational application. We also refer to~\cite{Fang_2021} for many further  quantities that may be rewritten as an optimization problem involving the relative entropy of channels such that the tools developed here become applicable.
 
\medskip 
\textbf{Acknowledgments}. Part of this work was completed during the conference ``Beyond IID in Information Theory,'' held at the University of Illinois Urbana--Champaign from July~29 to August~2, 2024, and supported by NSF Grant No.~2409823. We are grateful to the conference organizers for providing the environment that led to our collaboration.
GK thanks Marco Tomamichel for hosting him at CQT Singapore in fall 2024, and he thanks Mario Berta for hints and advice around the topic.
GK acknowledges funding by the European Research Council (ERC Grant Agreement No. 948139).
MMW acknowledges support from the National Science Foundation under Grant No.~2329662.

\section*{Author Contributions}

The following describes the different contributions of the authors of this work, using roles defined by the CRediT
(Contributor Roles Taxonomy) project~\cite{NISO}:

\medskip 
\noindent \textbf{GK}:
Conceptualization, Formal Analysis,  Investigation, Methodology, Software, Validation, Writing – original draft,  Writing - Review \& Editing.

\medskip 
\noindent \textbf{MMW}: Conceptualization, Formal Analysis,  Investigation, Methodology, Validation,  Writing – original draft, Writing - Review \& Editing.

\bibliography{references}
\bibliographystyle{IEEEtran}

\appendix

\section{Technical lemmas}

\begin{lemma}[\cite{Cooney_2016}]
\label{lem:optimization_restriction}
    For the channel relative entropy~\eqref{eq:problem_statement_channel_relative}, we have
    \begin{align}
        D(\mathcal{M}\Vert \mathcal{N}) = \sup_{\psi_{A A}} D(\mathcal{M}_{A\to B}(\psi_{A A}) \Vert \mathcal{N}_{A\to B}(\psi_{A A}))
    \end{align}
    where the optimization is over every pure state $\psi_{A A} \in \mathcal{S}(\mathcal{H}_{A} \otimes \mathcal{H}_A)$  and $\mathcal{H}_{A} \cong \mathcal{H}_A$.
\end{lemma}

\begin{proof}
    This was proven in~\cite[Lemma~7]{Cooney_2016} (see also~\cite[Definition~II.2]{Leditzky_2018} and the discussion thereafter). 
    Consider the general definition of relative entropy of channels
    \begin{align}
    D(\mathcal{M} \Vert \mathcal{N}) \coloneqq  \sup_{\rho_{AR}} D(\mathcal{M}_{A\to B}(\rho_{AR})\Vert \mathcal{N}_{A\to B}(\rho_{AR})).
\end{align}
For each $\rho_{RA}$ we can find a purification $\psi_{\tilde{R}RA}$ in an ancilla system $\mathcal{H}_{\tilde{R}}$. Then after applying the data-processing inequality, it is immediate that we can optimize over pure states $\psi_{RA}$, because we can absorb $\tilde{R}$ always in $R$. By the Schmidt decomposition of a pure state $\psi_{RA}$ it is clear that the number of Schmidt coefficients is  bounded from above by $\operatorname{dim}{dim}\mathcal{H}_A$. Thus we can optimize without loss of generality over states $\psi_{A A}$ and choose $\mathcal{H}_{A} \cong \mathcal{H}_A$. 
\end{proof}

\begin{lemma}[\cite{Cooney_2016}]
\label{lem:conic_formulation_relative_channels}
    For the channel relative entropy~\eqref{eq:problem_statement_channel_relative} we have 
    \begin{equation}
    \begin{aligned}
         D(\mathcal{M}\Vert \mathcal{N}) = \sup D(\rho_{A}^{1/2}&\Gamma^\mathcal{N}_{A B}\rho_{A}^{1/2} \Vert \rho_{A}^{1/2}\Gamma^\mathcal{M}_{A B}\rho_{A}^{1/2})  \\
         \rho_{A} &\in \mathcal{S}(\mathcal{H}_{A})
    \end{aligned}
    \end{equation}
    where $\Gamma^\mathcal{N}_{A B}$ and $\Gamma^\mathcal{M}_{A B}$ are the Choi matrices for the channels $\mathcal{M}$ and $\mathcal{N}$, respectively.
\end{lemma}

\begin{proof}
    The proof is from~\cite[Lemma $6$]{Cooney_2016}. Let $\psi_{A A}$ be a pure state. From cyclicity of the maximally entangled operator $ \Gamma_{A A}$ 
    we find under the Hilbert--Schmidt isomorphism a local operator $X_{A \to A} = I(\psi_{A A})$ acting just on $\mathcal{H}_{A}$ such that
    \begin{align}
        (X_{A\to A} \otimes \mathds{1}_{A} )\Gamma_{A A}(X_{A\to A} \otimes \mathds{1}_{A} )^{\dag} = \psi_{A A}.
    \end{align}
    Additionally the Hilbert--Schmidt isomorphism $I:\mathcal{H}_{A} \otimes \mathcal{H}_A \to \mathcal{B(\mathcal{H}_{A},\mathcal{H}_A})$ yields 
    \begin{align}
        \langle \psi |\phi\rangle _{A A} = \tr[I(\psi_{A A})^\dag I(\phi_{A A})] 
    \end{align}
    for all $\psi_{A A},\phi_{A A} \in \mathcal{H}_{A} \otimes \mathcal{H}_A$.
    Thus we have $\tr[X^\dag X] = 1$ and applying the polar decomposition of $X$ yields a unitary $U$ such that $X = U \vert X\vert$. Now we can start the calculation. Let 
    \begin{align}
        & D(\mathcal{M}_{A\to B}(\psi_{A A}) \Vert \mathcal{N}_{A\to B}(\psi_{A A}) ) \notag \\
        &= D(\mathcal{M}_{A\to B} (X_{A} \Gamma_{AA}  X^\dag_{A}) \Vert  \mathcal{N}_{A\to B} (X_{A} \Gamma_{AA}  X^\dag_{A}) )\\
        &= D(X_{A}\Gamma_{A B}^{\mathcal{M}}X_{A}^\dag \Vert X_{A}\Gamma_{A B}^{\mathcal{N}}X_{A}^\dag ).
    \end{align}
    Choosing now $\rho_{A} \coloneqq  X^\dag X$ leads to $X = U \rho_{A}^{1/2}$ to be the unitary multiplied by the unique positive square root. Using the unitary invariance of the relative entropy yields the desired result.
\end{proof}

%\onecolumn

\begin{lemma}\label{lem:dual_program}
Let $J_{AB}$ and $H_A$ be Hermitian operators, and let $E\in \mathbb{R}$. Suppose there exists a state
$\rho_{A}$ satisfying the energy constraint $\operatorname{tr}[H_{A}\rho
_{A}]\leq E$. Then
\begin{align}
& \sup_{\rho_{A}\geq0}\left\{  \operatorname{tr}[J_{AB}\left(  \rho_{A}\otimes
I_{B}\right)  ]:\operatorname{tr}[H_{A}\rho_{A}]\leq E,\ \operatorname{tr}
[\rho_{A}]=1\right\}  \notag 
\\& =\inf_{\substack{y\geq0,\\x\in\mathbb{R}}}\left\{  x+yE:xI_{A}+yH_{A}\geq\operatorname{tr}
_{B}[J_{AB}]\right\}
\end{align}

\end{lemma}

\begin{proof}
Consider that
\begin{align}
&  \sup_{\rho_{A}\geq0}\left\{  \operatorname{tr}[J_{AB}\left(  \rho
_{A}\otimes I_{B}\right)  ]:\operatorname{tr}[H_{A}\rho_{A}]\leq
E,\ \operatorname{tr}[\rho_{A}]=1\right\}  \nonumber\\
&  =\sup_{\rho_{A}\geq0}\!\left\{  \operatorname{tr}[J_{AB}\!\left(  \rho
_{A}\otimes I_{B}\right)  ]+\inf_{\substack{y\geq0,\\x\in\mathbb{R}}} \left\{ \begin{array}{c}
       y\left(
E-\operatorname{tr}[H_{A}\rho_{A}]\right) \\
      +x\left(  1-\operatorname{tr}
[\rho_{A}]\right)
\end{array}   \right\} \right\}  \notag \\
&  =\sup_{\rho_{A}\geq0}\inf_{\substack{y\geq0,\\x\in\mathbb{R}}}\left\{  \begin{array}{c}
       \operatorname{tr}
[J_{AB}\left(  \rho_{A}\otimes I_{B}\right)  ]+y\left(  E-\operatorname{tr}
[H_{A}\rho_{A}]\right) \\
      +x\left(  1-\operatorname{tr}[\rho_{A}]\right)
\end{array} 
\right\} \notag  \\
&  =\sup_{\rho_{A}\geq0}\inf_{\substack{y\geq0,\\x\in\mathbb{R}}}\left\{
x+yE+\operatorname{tr}[\left(  \operatorname{tr}_{B}[J_{AB}]-xI_{A}
-yH_{A}\right)  \rho_{A}]\right\} \notag  \\
&  \leq\inf_{\substack{y\geq0,\\x\in\mathbb{R}}}\sup_{\rho_{A}\geq0}\left\{
x+yE+\operatorname{tr}[\left(  \operatorname{tr}_{B}[J_{AB}]-xI_{A}
-yH_{A}\right)  \rho_{A}]\right\} \notag  \\
&  =\inf_{\substack{y\geq0,\\x\in\mathbb{R}}}\!\left\{  x+yE+\sup_{\rho_{A}\geq0}\left\{
\operatorname{tr}[\left(  \operatorname{tr}_{B}[J_{AB}]-xI_{A}-yH_{A}\right)
\rho_{A}]\right\}  \right\} \notag  \\
&  =\inf_{\substack{y\geq0,\\x\in\mathbb{R}}}\left\{  x+yE:xI_{A}+yH_{A}\geq
\operatorname{tr}_{B}[J_{AB}]\right\}  .
\end{align}
A feasible choice for the maximization problem is $\rho_{A}=|\psi
\rangle\!\langle\psi|_{A}$, where $|\psi\rangle$ is the eigenvector
corresponding to the minimum eigenvalue of $H$. Fix $\varepsilon,\delta>0$. A
strictly feasible choice for the minimization problem is
\begin{align}
y &  =\varepsilon , \\
x &  =\left\vert \lambda_{\max}\left(  \operatorname{tr}_{B}[J_{AB}]-yH_{A}\right)
\right\vert +\delta .
\end{align}
Thus, strong duality holds, and the inequality above is saturated.
\end{proof}

\end{document}